\documentclass[11pt]{article}

\usepackage{todonotes}
\usepackage{tikz}
\usepackage{tkz-euclide}
\usepackage{algorithm}
\usepackage[noend]{algpseudocode}
\usepackage{authblk}
\usepackage{amsmath,amssymb,amsthm}

\usepackage{hyperref}
\usepackage{thmtools}
\hypersetup{colorlinks=true,linkcolor=blue,filecolor=blue,citecolor=blue,urlcolor=blue}
\declaretheorem[name=Theorem]{theorem}
\declaretheorem[name=Lemma, sibling=theorem]{lemma}
\declaretheorem[name=Proposition, sibling=theorem]{proposition}
\declaretheorem[name=Definition, sibling=theorem]{definition}
\declaretheorem[name=Corollary, sibling=theorem]{corollary}

\declaretheorem[name=Example, sibling=theorem]{example}

\newcommand{\G}{\mathcal{G}}

\title{Designing sparse temporal graphs satisfying connectivity requirements} 

\author[1]{Thomas Bellitto}

\author[1]{Jules Bouton Popper}

\author[1,2]{Justine Cauvi}

\author[1]{Bruno Escoffier}

\author[1]{Raphaëlle Maistre-Matus}

\affil[1]{Sorbonne Universit\'e, CNRS, LIP6, F-75005 Paris, France}

\affil[2]{\'Ecole Normale Supérieure de Lyon, Lyon, France}

\date{}

\begin{document}

\maketitle

\begin{abstract}
Connectivity of temporal graphs has been widely studied both as graph theory and as gossip theory. In particular, it is well known that in order to connect every vertex to every other, a temporal graph needs to have at least $2n-4$ edges where $n$ is the number of vertices. This paper investigates the optimal number of edges required to satisfy partial connectivity requirements. We introduce the problem of Connectivity Request Satisfaction where we are given a directed graph that we call the request graph, where an arc from $u$ to $v$ means that we need to be able to go from $u$ to $v$. Our goal is to build a temporal graph on the same vertex set with as few temporal edges as possible that would satisfy all the requests. When the graph we build is directed, we prove that the number of temporal arcs required is $n-\mathrm{cc}+\mathrm{dfvs}$ where $\mathrm{cc}$ is the number of connected component of the request graph and $\mathrm{dfvs}$ is the size of its smallest directed feedback vertex set. It follows that the problem is NP-complete but inherits fixed parameter tractability properties of Directed Feedback Vertex Set. When the graph we build is undirected, we establish a characterization of strongly connected request graphs that admit a solution with $n-1$ edges: it is possible if and only if any set of pairwise non-vertex-disjoint closed walks all share a common vertex. We prove that this criteria can be tested in polynomial time.
 
\end{abstract}

\section{Introduction}
\label{sec:introduction}

A temporal graph is a powerful tool to model and analyse complex real-world networks such as public transit networks, phone call networks, social networks or biological networks. A temporal graph allows for edges to change over time. More precisely, it is defined as a graph with each edge being equipped with a set of appearance times. Connectivity in this context is defined by journeys, that is a walk where edges are traversed one after the other in time. (Shortest) paths, connectivity and exploration problems have been at the core of the study of temporal graphs ~\cite{bellitto2026,buixuan2003,bumpus2023,dreyfus1969,kempe2002,othon2016,wu2016}. While some notions and algorithms quite easily transfer from static to temporal graphs, some other notions (such as spanning trees) and problems turn out to be much more complex in temporal graphs. Among those, a famous result in temporal graphs concerns the number of edges in a connected temporal graphs. In static graphs, minimal (in term of number of edges) graphs that connect a set of $n$ vertices are trees, which have $n-1$ edges. In temporal graphs, the question is more complex: how many (temporal) edges do we need to connect a set of $n$ vertices? This problem has been introduced in the literature as the {\it gossip problem} in the context of optimal scheduling of phone calls (see e.g.~\cite{Hedetniemi1988} for a survey on gossiping), and it is known that for $n\geq 4$ the minimum number of edges is $2n-4$ (see Figure~\ref{fig:2n-4} for an illustration)~\cite{Bumby1981,Hajnal1972}. Most of the works in that field consider that at most one phone call can be done between two persons. It was shown in~\cite{Gobel1991} that deciding if a given graph is label-connected, that is every edge can be given an appearance time such that there is a journey between every pair of vertices, is NP-complete and a characterization of label-connected graphs with $2n-4$ edges was proposed. 

\begin{figure}[t]
\begin{center}
\begin{tikzpicture}[scale = 2]
			\begin{scope}

                \draw[black, fill = black] (0,0) circle (.03);

                \draw[black, fill = black] (1,0) circle (.03);
                
                \draw[black, fill = black] (1.4,0) circle (.01);
                \draw[black, fill = black] (1.5,0) circle (.01);
                \draw[black, fill = black] (1.6,0) circle (.01);
                
                \draw[black, fill = black] (2,0) circle (.03);

                \draw[black, fill = black] (3,0) circle (.03);
                
                \draw[black, fill = black] (4,0) circle (.03);

                \draw[black, fill = black] (5.5,0) circle (.03);
                
                \draw[black, fill = black] (4.75,0.75) circle (.03);

                \draw[black, fill = black] (4.75,-0.75) circle (.03);
                
                \draw [black, -] plot [smooth] coordinates {(4.75,0.75) (4,0.5) (3,0)};
                
                \draw [black, -] plot [smooth] coordinates {(4.75,0.75) (3,0.5) (2,0)};
                
                \draw [black, -] plot [smooth] coordinates {(4.75,0.75) (2,0.5) (1,0)};
                
                \draw [black, -] plot [smooth] coordinates {(4.75,0.75) (1,0.5) (0,0)};

                \draw[-] (0,0) to (1,0);
                
                \draw[-] (2,0) to (4,0);
                
                \draw[-] (4.75,-0.75) to (4,0);
                
                \draw[-] (4.75,0.75) to (4,0);
                
                \draw[-] (4.75,-0.75) to (5.5,0);
                
                \draw[-] (4.75,0.75) to (5.5,0);

                \tkzDefPoint(0.5,0){0}
				\tkzLabelPoint[below](0){$1$}
				
				\tkzDefPoint(2.5,0){0}
				\tkzLabelPoint[below](0){$n-5$}
				
				\tkzDefPoint(3.5,0){0}
				\tkzLabelPoint[below](0){$n-4$}
				
				\tkzDefPoint(4.35,-0.38){0}
				\tkzLabelPoint[left](0){$n-1$}
				
				\tkzDefPoint(4.35,0.38){0}
				\tkzLabelPoint[left](0){$n-3$}
				
				\tkzDefPoint(5.2,-0.38){0}
				\tkzLabelPoint[right](0){$n-2$}
				
				\tkzDefPoint(5.2,0.38){0}
				\tkzLabelPoint[right](0){$n$}
				
				\tkzDefPoint(3.2,0.2){0}
				\tkzLabelPoint[above](0){$n+1$}
				
				\tkzDefPoint(2.2,0.2){0}
				\tkzLabelPoint[above](0){$n+2$}
				
				\tkzDefPoint(1.15,0.2){0}
				\tkzLabelPoint[above](0){$2n-5$}
				
				\tkzDefPoint(0.1,0.2){0}
				\tkzLabelPoint[above](0){$2n-4$}

	\end{scope}
\end{tikzpicture}
\end{center}
\caption{A solution to the gossip problem with $2n-4$ edges. A label next to an edge corresponds to the appearance time of the edge. Every vertex can reach every other vertex by a path that uses edges with increasing appearance time.}
\label{fig:2n-4}
\end{figure}
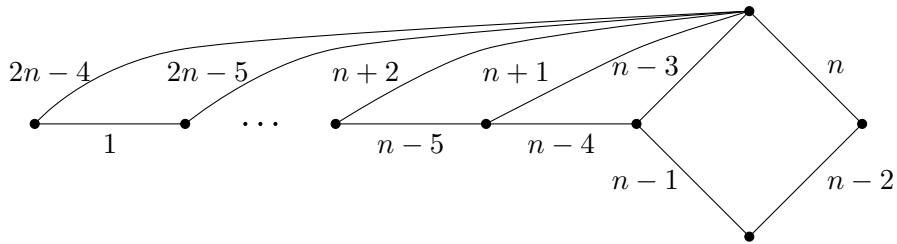

However, if we only need to connect some vertices of the graph to some others, we would expect to be able to do it at a smaller cost. Still, while this can be relevant for many practical applications of temporal graphs, very little is known about it. We thus introduce a generalization of the gossip problem where we do not require full connectivity of the graph but instead we are given as input a set of ordered pairs $(u,v)$ that we have to connect. These connectivity requests can be seen as a directed graphs where an arc $(u,v)$ means that we need to be able to go from $u$ to $v$ in the temporal graph $\mathcal G$ that we create. Hence, the gossip problem is the case where the input is a complete graph.\\

\noindent {\bf Contributions}

As explained above, given a directed graph $R=(V,A)$ called the request graph, we want to design a temporal graph $\mathcal{G}$ on the same set of vertices, such that for every $(u,v)$ in $A$ there is a journey from $u$ to $v$ in $\mathcal{G}$. We first consider the directed version of the problem, where we want to design a directed temporal graph. We establish a strong link between our problem and the well known feedback vertex set problem. More precisely, we show that the minimal number of temporal edges we need to satisfy all the connectivity requirements is $n-cc+f$, where $n=|V|$, $cc$ is the number of connected components of $R$, and $f$ is the size of a minimum feedback vertex set in $R$. This characterization shows in particular that the problem is NP-hard.

We then tackle the problem in undirected temporal graphs. This is the case of the gossip problem mentioned above, which states that $2n-4$ edges is the minimum number of temporal edges when $R$ is a complete graph. On the other hand, it is easy to see that $n-cc$ edges are always necessary (where $cc$ is the number of connected components of $R$). We tackle the question to determine when this lower bound is reached, i.e., when the minimum number of temporal edges needed to satisfy the connectivity requirements is equal to $n-cc$. We provide a characterization when $R$ is strongly connected: in this case, we show that there exists a solution which is a tree if and only if $R$ satisfies a property that we call walk-Helly, which is a Helly-like property on walks in graphs. We also show a polynomial time algorithm that builds such a tree when it exists. This easily extends to the case where every connected components of $R$ are strongly connected.

{\it Organization:} Section~\ref{sec:preliminaries} formally introduces the problems we consider in this article. The directed case is dealt with in Section~\ref{sec:directed}, while the undirected case is studied in Section~\ref{sec:undirected}. Some open questions are given in the conclusion in Section~\ref{sec:conclusion}.\\ 

\noindent {\bf Related works}

Graph realization problems have been studied on static graphs since the 1960s and consist of finding a graph that satisfy a given property $P$ or answering no if such a graph does not exist~\cite{erdos1960,hakimi1965}. It was recently introduced in the context of temporal graphs in~\cite{klobas2024} and is an active topic of research~\cite{casteigts2025,cauvi2025,erlebach2025,erlebach2024,mertzios2025,mertzios2025-2,meusel2025}. In particular, the authors in~\cite{erlebach2025} studied the Reachability Graph Realizability problem that asks if the input directed graph is the reachability graph of some temporal graph, the reachability graph of a temporal graph $\G$ being the directed graph for which there is an arc from $u$ to $v$ if and only if there is a journey from $u$ to $v$ in $\G$. They studied the problem with various restrictions on the labeling of the temporal graph and showed that the problem is NP-hard for most of the variants, among other results. Note that our problem is different from the Reachability Graph Realizability problem as in our problem we can have a journey from $u$ to $v$ even if the arc $(u,v)$ is not in the request graph. Forbidding journeys plays a key role in the results on Reachability Graph Realizability, and the structures of solutions of the two problems differ fundamentally. 
In realization problems, the problem of constructing a temporal graph realizing the given property $P$ is a particular instance of the temporal network design problem in which one aims at building a temporal graph satisfying some constraints while optimizing some measures. Many temporal network design problems have been studied (see e.g.~\cite{akrida2017,deligkas2022,enright2021,klobas2024-2,mertzios2013}). 

Our work is also closely related to the study of temporal spanners, introduced in \cite{kempe2002}, which are minimal connected subgraph of a given temporal graph. Spanners have also been considered for weaker definitions of connectivity in \cite{kurita_et_al:LIPIcs.SAND.2025.9}.

Finally, our work can be seen as a specific case of the Temporal Pair Connectivity Augmentation Problem introduced in \cite{bellitto_et_al:LIPIcs.SAND.2025.3}. One of the problems the authors considered is, given a temporal graph and a set of unsatisfied connectivity requests, to satisfy the constraints by adding as few edges as possible to the input graph. Our problem here can be seen as the specific case where the input graph has no edge. The authors of \cite{bellitto_et_al:LIPIcs.SAND.2025.3} proved several hardness results for their problem, but their proofs use sophisticated input graphs and do not prove that the problem is difficult even when restricted to empty graphs.

\section{Preliminaries}
\label{sec:preliminaries}

\subsection{Basic definitions and notations}

For natural numbers $i\leq j$, we note $[i,j]:=\{i,i+1,\dots,j\}$.

We will consider both undirected graphs and directed graphs (digraphs).  For an undirected graph $G=(V,E)$  (resp. for a digraph $G=(V,A)$) and $S\subseteq V$ a vertex set, we define the {\it subgraph of $G$ induced by $S$} as $G[S]=(S, \{\{u,v\}\in E\:|\: u,v\in S\})$ (resp. $G[S]=(S, \{(u,v)\in A\:|\: u,v\in S\})$). 

A {\it walk} in an undirected graph (resp. digraph) $G$  between $u_0$ and $u_k$ (resp. from $u_0$ to $u_k$) is a sequence $W=(u_0,u_1,\dots,u_k)$ of vertices such that $\{u_i,u_{i+1}\}$ is an edge (resp. $(u_i,u_{i+1})$ is an arc) for $i\in [0,k-1]$. If $u_k=u_0$, we say that the walk $W$ is {\it closed}. If the vertices $u_0,\dots,u_k$ are distinct, it is called a {\it path}. If the vertices are distinct except for $u_0=u_k$ it is called a {\it cycle} (resp. {\it a circuit}).

We say that an undirected graph (resp. a digraph) is {\it connected} (resp. {\it strongly connected}) if there is a path between $s$ and $t$ (resp. from $s$ to $t$) for every vertex pair (resp. ordered pair) $s,t$.
A {\it connected component} (resp. {\it strongly connected component}) of an undirected graph (resp. of a digraph) $G$ is a connected (resp. strongly connected) subgraph of $G$ induced by a vertex set $S$ such that there is no $S'$ with $S\subset S'$ and $G[S']$ connected (resp. strongly connected).

For a digraph $G=(V,A)$, we define the underlying undirected graph $G_{\text{undir}}=(V,E)$ with $E=\{\{u,v\}\:|\:(u,v)\in A\}$. We say that a digraph $G$ is {\it connected} if $G_{undir}$ is connected. A {\it connected component} of a digraph $G$ is a connected subgraph of $G$ induced by a vertex set $S$ such that there is no $S'$ with $S\subset S'$ and $G[S']$ connected. 

\subsection{Temporal graphs}

\begin{definition}[Temporal graph]
        A temporal graph is a pair $\mathcal{G}=(V,\mathcal{E})$ where $V$ is a (finite) set of vertices and $\mathcal{E}=\{e_1,\dots,e_m\}$ is a set of (distinct) temporal edges. A temporal edge $e_i$ is a pair $(\{u,v\},t)$ where $u,v$ are distinct vertices of $V$ and $t\in \mathbb{N}^*$. $u$ and $v$ are called endpoints of $e_i$ and $t$ the appearance time of $e_i$.
\end{definition}    
    Note that there might be several temporal edges with the same endpoints (but with different appearance times). A temporal graph is called {\it simple} if for any pair $\{u,v\}$ there is at most one temporal edge with endpoints $u$ and $v$. 

    The {\it snapshot} of $\mathcal{G}=(V,\mathcal{E})$ at time $t$ is the (static) graph $G_t=(V,E_t)$ where $E_t$ is the set of edges whose appearance time is $t$. The {\it footprint} of $\G$ is the (static) graph $\G_{\downarrow}=(V,\mathcal{E}_{\downarrow})$ with $\mathcal{E}_{\downarrow}=\{\{u,v\}\:|\:(\{u,v\},t)\in \mathcal{E}\}$.
    
    {\it Temporal digraphs} are defined similarly, on a vertex set $V$ and on a set of temporal arcs where a temporal arc is a pair $e=((u,v),t)$ (the arc is directed from $u$ to $v$).  
    
    \begin{definition}[Journeys]
        A \textit{journey} of a temporal graph (resp. a temporal digraph) $\mathcal{G}$ is a sequence  $(u_0,u_1,t_0),(u_1,u_2,t_1),\dotsc,(u_{k-1},u_k,t_{k-1})$ where:
        \begin{itemize}
            \item $(\{u_i,u_{i+1}\},t_i)$ is a temporal edge (resp. $((u_i,u_{i+1}),t_i)$ is a temporal arc) of $\mathcal{G}$;
            \item $t_i < t_{i+1}$.
        \end{itemize}
    \end{definition} 
    Journeys can be also written as sequences of vertices (as in the case of static graphs) together with the time of the temporal edge/arc linking two consecutive vertices.
    
    Note that this definition corresponds to what is usually called a {\it strict journey} (as time is required to be strictly increasing), as opposed to a non-strict journey (where time is only required to be non-decreasing). As we will only focus on strict journeys (for reasons that we explain at the end of the next subsection), we omit to precise {\it strict} in the article. 
    
    \begin{definition}[Reachability]
        In a temporal graph $\mathcal{G}=(V,\mathcal{E})$, a vertex $v$ is \textit{reachable} from $u$ if there exists a journey from $u$ to $v$. The reachability graph of $\mathcal{G}$ is the graph $\mathrm{Reach}(\mathcal{G})$ on vertex set $V$ containing all arcs $(u,v)$ such that $v$ is reachable from $u$ in $\mathcal{G}$. If every vertex $v$ is reachable from every other vertex $u$ in $\mathcal{G}$, i.e., if $\mathrm{Reach}(\mathcal{G})$ is complete, then $\mathcal{G}$ is \textit{temporally connected}. 
    \end{definition}

\subsection{Connectivity Request Satisfaction problems}

    Now we can formally define the main problems under consideration.
    \begin{quote}
        \textsc{Connectivity Request Satisfaction} (CRS)\\
        \textbf{Input:} A static digraph $R=(V,A)$, and integer $k$.\\
        \textbf{Question:} Is there a temporal graph $\mathcal{G}=(V,\mathcal{E})$ with at most $k$ temporal edges such that for all $(v_i,v_j)\in A$ there exists a journey from $v_i$ to $v_j$ in $\mathcal{G}$? 
    \end{quote}
    
    Equivalently, we want that the reachability graph of $\mathcal{G}$ contains all arcs of $R$. Note that the reachability graph of $\G$ can contain arcs that are not arcs of $R$.

    We define similarly the \textsc{Directed Connectivity Request Satisfaction} (DCRS) where the question is to determine whether there exists a temporal digraph fulfilling the connectivity requirements.

    \begin{quote}
        \textsc{Directed Connectivity Request Satisfaction} (DCRS)\\
        \textbf{Input:} A static digraph $R=(V,A)$, and integer $k$.\\
        \textbf{Question:} Is there a temporal digraph $\mathcal{G}=(V,\mathcal{E})$ with at most $k$ temporal arcs such that for all $(v_i,v_j)\in A$ there exists a journey from $v_i$ to $v_j$ in $\mathcal{G}$? 
    \end{quote}
    
    When considering an instance $(R,k)$ of CRS (resp. DCRS), we call $R$ the {\it request graph}. In the following, we will also consider the minimization problem, which we call \textsc{Min Connectivity Request Satisfaction} (MinCRS) (resp. \textsc{Directed Connectivity Request Satisfaction} (MinDCRS)), in which only a request graph $R$ is given as input and one asks what is the minimum number of temporal edges (resp. temporal arcs) of a temporal graph (resp. of a temporal digraph) whose reachability graph contains all arcs of $R$. 

    \begin{example}\label{ex:example}
    Let us consider the request graph $R$ depicted in Figure~\ref{fig:ex} (left part). In the center part of of Figure~\ref{fig:ex} is given a feasible (in fact, optimal) solution $\mathcal G$ for MinCRS. For instance, in $\mathcal G$ there is a journey from $c$ to $a$, thus fulfilling the request $(c,a)$ in $R$. We note that in this example the reachability graph of $\mathcal G$ strictly contains $R$, as for instance there is a journey from $b$ to $d$. More precisely, the reachability graph of $\mathcal G$ contains all arcs but $(d,a)$, $(d,b)$ and $(e,b)$.
    
    On the right part is given a feasible (in fact, optimal) solution $\mathcal G'$ of MinDCRS. It uses 6 arcs, while the optimal solution of MinCRS uses only 5 edges (using the fact that edges can be taken in both directions in a journey).

    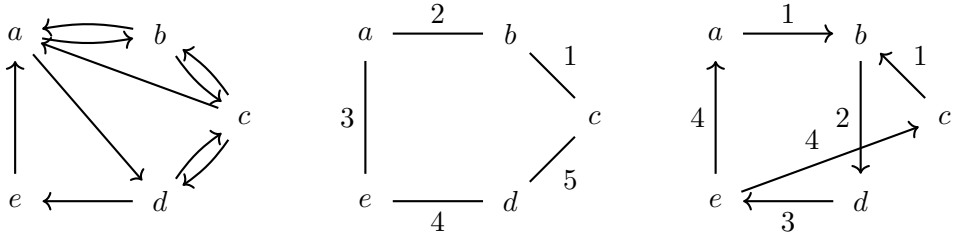
\begin{figure}[!ht]
      \centering
        \begin{tikzpicture}[node distance={20mm}, thick, main/.style = {draw=none, circle, minimum size=0.7cm}]
          \node[main] (a) [] {$a$};
          \node[main] (b) [right = 1.2cm of a] {$b$};
          \node[main] (c) [below right= 0.6cm and 0.6cm of b] {$c$};
          \node[main] (d) [below left = 0.6cm and 0.6cm of c] {$d$};
          \node[main] (e) [left = 1.2cm of d] {$e$};
          \draw[->] (c) to[out=-125,in=35] (d);
          \draw[->] (d) to[out=55,in=-145] (c);
          \draw[->] (d) -- (e);
        \draw[->] (a) -- (d);
                \draw[->] (c) -- (a);
          \draw[->] (e) -- (a);
        \draw[->] (b) to[out=170,in=10] (a) ;
        \draw[->] (a) to[out=-10,in=-170] (b) ;
         \draw[->] (b) to[out=-55,in=145] (c) ;
        \draw[->] (c) to[out=125,in=-35] (b) ;
        
          \node[main] (a2) [right = 2cm of b] {$a$};
          \node[main] (b2) [right = 1.2cm of a2] {$b$};
          \node[main] (c2) [below right= 0.6cm and 0.6cm of b2] {$c$};
          \node[main] (d2) [below left = 0.6cm and 0.6cm of c2] {$d$};
          \node[main] (e2) [left = 1.2cm of d2] {$e$};
          \draw[-] (a2) -- node[midway, above] {2} (b2);
          \draw[-] (b2) -- node[midway, above right] {1} (c2);
          \draw[-] (c2) -- node[midway, below right] {5} (d2);
            \draw[-] (d2) -- node[midway, below] {4} (e2);
          \draw[-] (e2) -- node[midway, left] {3} (a2);

          \node[main] (a3) [right = 2cm of b2] {$a$};
          \node[main] (b3) [right = 1.2cm of a3] {$b$};
          \node[main] (c3) [below right= 0.6cm and 0.6cm of b3] {$c$};
          \node[main] (d3) [below left = 0.6cm and 0.6cm of c3] {$d$};
          \node[main] (e3) [left = 1.2cm of d3] {$e$};
          \draw[->] (a3) -- node[midway, above] {1} (b3);
          \draw[->] (c3) -- node[midway, above right] {1} (b3);
          \draw[->] (b3) -- node[midway, left] {2} (d3);
            \draw[->] (d3) -- node[midway, below] {3} (e3);
            \draw[->] (e3) -- node[midway, above left] {4} (c3);
          \draw[->] (e3) -- node[midway, left] {4} (a3);
        \end{tikzpicture}
      \caption{Request graph $R$ (left), optimal solutions for MinCRS (center part) and for MinDCRS (right part).}
      \label{fig:ex}
    \end{figure}
\end{example}

    Note that if we were considering non strict paths/walks in temporal graphs, then in both CRS and DCRS we would use only one time step (all the edges/arcs would have the same appearance time), and then the problems would be equivalent to the problems in static graphs. In static graphs, the problems are trivial: in the non directed case, an optimal solution is to build a tree on every connected component of the request graph. In the directed case, an optimal solution is to build a circuit in all connected components of the request graph that contains a circuit and a path on the others.

\section{Directed case}
\label{sec:directed}
  In this section, we study DCRS. We show that the problem is equivalent to the minimum directed feedback vertex set problem. To this end, we first present a construction that uses $n+f-cc$ temporal edges, where $n=|V|$ is the number of vertices of $R=(V,A)$, $cc$ its number of connected components and $f$ the size of a directed feedback vertex set in $R$. Then we will prove that this construction is minimal, i.e. there is no solution to DCRS if $k$ is smaller.
  \begin{definition}[Directed Feedback Vertex Set (DFVS)]
    Let $G=(V,A)$ be a digraph. $S \subseteq V$ is a DFVS if every circuit of $G$ contains a vertex in $S$.
  \end{definition}
  \begin{lemma}
    \label{lemma:constr}
    Let $R$ be a request graph. There is a temporal graph satisfying $R$ with $n+f-cc$ temporal edges.
  \end{lemma}
  \begin{proof}
    We construct a solution with $n+f-1$ arcs when the graph $R$ is connected. Let $F=\{v_1,\dots,v_f\}$ be a DFVS of $R$, and $(u_1,\dots,u_{n-f})$ be a topological order induced by $R$ after the removal of $F$. At time 1, we have arcs from the DFVS to $u_1$: $$\mathcal E_1 = \{((v_i,u_1),1)| i \in \{1,\dots,f\}\}$$
    We then add the temporal arcs along the topological order:
    $$\mathcal E_2=\{((u_1,u_2),2),\dots,((u_{n-f-1},u_{n-f}),n-f)\}$$
    Finally we add arcs from $u_{n-f}$ to the DFVS: 
    $$\mathcal E_3 = \{((u_{n-f},v_i),n-f+1)| i \in \{1,\dots,f\}\}$$
    The construction is illustrated in Figure~\ref{fig:soldirected}. In $R$, any arc $a$ with an endpoint in the DFVS is satisfied in $\mathcal G = (V, \mathcal E_1 \cup \mathcal E_2 \cup \mathcal E_3)$ because there is a journey from any vertex of the DFVS to any other vertex in $\mathcal G$ and vice versa. Since every remaining arc is oriented in the same direction as the topological order, the request is filled along the path $\mathcal E_2$. This connected construction has $|\mathcal E_1|+|\mathcal E_2|+|\mathcal E_3|=2f+n-f-1=n+f-1$ temporal edges.  This can be extended to $R$ with multiple connected components, taking a DFVS of each component and having the same construction but on this DFVS and the vertices of that component. The total number of edges would then be $n+f-cc$.
    \begin{figure}[!ht]
      \centering
        \begin{tikzpicture}[node distance={20mm}, thick, main/.style = {draw=none, circle, minimum size=0.9cm}]
          \node[main] (vf) [] {$v_f$};
          \node[main] (v1) [below = 2cm of vf] {$v_1$};
          \node[main] (u1) [below right= 1cm and 2cm of vf] {$u_1$};
          \node[main] (u2) [below = 1cm of u1] {$u_2$};
          \node[main] (unf1) [left = 4cm of u2] {$u_{n-f-1}$};
          \node[main] (unf) [above = 1cm of unf1] {$u_{n-f}$};
          \node (intu) [draw=none,left = 1.75cm of u2] {\textbf{$\cdots$}};
          \node (intv) [draw=none, below = 0.5cm of vf] {\textbf{$\vdots$}};
          \draw[->] (vf) -- (u1);
          \draw[->] (intv) -- node[midway, below] {1} (u1);
          \draw[->] (v1) -- (u1);
          \draw[->] (u1) -- node[midway, right] {2} (u2);
          \draw[->] (u2) -- node[midway, below] {3} (intu);
          \draw[->] (intu) -- node[midway, below] {\small{$n-f-1$}} (unf1);
          \draw[->] (unf1) -- node[midway, left] {$n-f$} (unf);
          \draw[->] (unf) -- (vf);
          \draw[->] (unf) -- (v1);
          \draw[->] (unf) -- node[midway, below] {$n-f+1$} (intv);
        \end{tikzpicture}
      \caption{Construction for a connected component $C$, with $n$ vertices and a DFVS of size $f$, of a request graph $R$ using $n+f-1$ temporal arcs and satisfying the requests of $R$ in $C$.}
      \label{fig:soldirected}
    \end{figure}
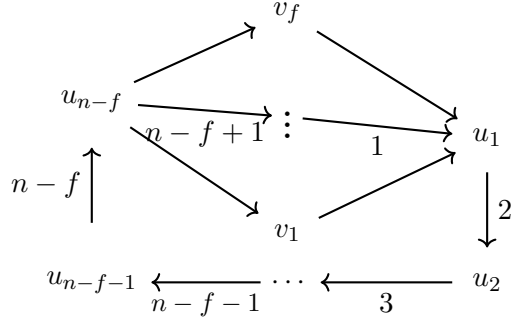
  \end{proof}
  In order to prove the minimality of this construction, we first prove the following lemmas:
  \begin{lemma}
  \label{lemma:sccfootprintreach}
    Let $\mathcal{G}$ be a temporal digraph, $\mathrm{Reach}(\mathcal G)$ and $\mathcal G_{\downarrow}$ have the same strongly connected components.
  \end{lemma}
  \begin{proof}
    \begin{itemize}
      \item Suppose $S \subseteq V$ is an SCC of $\mathrm{Reach}(\mathcal G)$. Then for $u,v \in S$ in $\mathcal G$ there is a sequence of directed journeys between $u$ and $v$ in both directions by definition of $\mathrm{Reach}(\mathcal G)$, which imply that there exists a path from $u$ to $v$ and a path from $v$ to $u$ in $\mathcal G_{\downarrow}$. Thus $S \subseteq S'$ with $S'$ an SCC in $\mathcal G_{\downarrow}$.
      \item Suppose $S \subseteq V$ is an SCC of $\mathcal G_{\downarrow}$. Then for $u,v \in S$, we have the directed paths from $u$ to $v$ and $v$ to $u$ in $\mathcal G_{\downarrow}$, and each arc of those paths are also in $\mathrm{Reach}(\mathcal G)$. Thus $S \subseteq S'$ with $S'$ an SCC of $\mathrm{Reach}(\mathcal G)$.\qedhere
    \end{itemize}
  \end{proof}
  
  We then highlight the relation between the size of a feedback vertex set and the number of temporal edges.
  
  \begin{lemma}
  \label{lemma:fvs}
  Let $\mathcal G=(V,\mathcal{E})$ be a temporal digraph such that $\mathcal G_{\downarrow}=(V,A)$ has $cc$ connected components. For $p = |\mathcal E| + cc - n$ there exists $S \subseteq V$ a DFVS of $\mathrm{Reach}(\mathcal G)$ with $|S| \leq p$.
  \end{lemma}
  \begin{proof}
    We proceed by induction on $p$. We note that by definition of the footprint $\mathcal G_{\downarrow}$ we have $|A|\leq |\mathcal{E}|$. In particular, $p$ cannot be negative (since $cc\geq n-|A|$, as in any graph). 
    \begin{itemize}
      \item For $p=0$, $|\mathcal{E}|=n-cc$. As $|A|\leq |\mathcal{E}|$, necessarily every connected component  $C$ of $\mathcal G_{\downarrow}$ is a tree. Thus $C$ is a DAG. By Lemma~\ref{lemma:sccfootprintreach}, we also have that $\mathrm{Reach}(\mathcal G)$ is a union of DAGs thus there exists a DFVS of size $p=0$.
      \item Suppose the property holds for some $p_0\geq 0$, let us prove it holds for $p=p_0+1$. Let $\mathcal G$ be a temporal digraph such that $\mathcal G_{\downarrow}=(V,A)$ has $cc$ connected components and $|\mathcal{E}|=n-cc+p$. If $\mathcal G_{\downarrow}$ has no circuit, then as in the case $p=0$ there is a DFVS of size $0\leq p$. Otherwise, let $C$ be a connected component of $\mathcal G_{\downarrow}$ which contains an SCC $S$ with at least 2 vertices. Finally, let us consider $((u,v),t_0) \in \mathcal E$ to be the temporal edge with both endpoints in $S$ with the smallest appearance time.
        
        We define $\mathcal G'$ to be the temporal digraph obtained from $\mathcal G$ by removing  $((u,v),t_0)$. As both $u$ and $v$ are in the SCC $S$, there is a path from $v$ to $u$ in $\mathcal G_\downarrow$, so $\mathcal G_\downarrow$ and $\mathcal G'_\downarrow$ have the same connected components. In particular, $\mathcal G'_\downarrow$ has $cc$ connected components. So we can apply our induction hypothesis to $\mathcal G'$ (which has one temporal arc less than $\mathcal G$) with $p_0=|\mathcal E|-1+n-cc$: there exists a DFVS $F'$ of $\mathrm{Reach}(\mathcal G')$ such that $|F'| \leq p_0$.
    
    Let us show that $F' \cup \{u\}$ is a DFVS of $\mathrm{Reach}(\mathcal G)$. Suppose a contrario that there exists a circuit $W$ in $\mathrm{Reach}(\mathcal G)$ that does not intersect $F' \cup \{u\}$. According to Lemma~\ref{lemma:sccfootprintreach}, the vertices of $W$ must be included in an SCC of $\mathcal G_\downarrow$. This SCC must be the one containing $u$ and $v$: indeed, in $\mathcal G_\downarrow$ there is no directed path between two vertices of $W$ using $(u,v)$ (otherwise $u$ and $v$ would be in that SCC), hence $W$ would be a circuit in $\mathrm{Reach}(\mathcal G')$, contradicting the definition of $F'$. 
    
    So the vertices of $W$ are in the SCC of $u$ and $v$. As $F'$ is a DFVS of $\mathcal G'$, there must be an arc $(a,b)$ in the circuit $W$ which is in $\mathrm{Reach}(\mathcal G)$ but not in $\mathrm{Reach}(\mathcal G')$. Hence, we have a journey from $a$ to $b$ in $\mathcal G$ that uses $((u,v),t_0)$. But by minimality of $t_0$ this path must start from $u$, a contradiction as $W$ does not contain $u$.  
    $F' \cup \{u\}$ is therefore a DFVS of $\mathrm{Reach}(\mathcal G)$ of size $|F'| + 1 \leq p_0 + 1=p$.\qedhere
    \end{itemize}
  \end{proof}
  
  We now prove the main theorem of this section:
  \begin{theorem}
\label{thm:maindirected}
    Let $(R,k)$ be an instance of DCRS with $cc$ connected components, and $f$ be the size of a minimum DFVS of $R$. There exists $\mathcal G$ a solution to this instance iff $k \geq n + f - cc$.
  \end{theorem}
  \begin{proof}
    Suppose that there exists a $\mathcal G$ solution such that $|\mathcal E| = n+f-cc-d$, with $d > 0$. Then by Lemma~\ref{lemma:fvs} we have a DFVS of $\mathrm{Reach}(\mathcal G)$ of size $f'=|\mathcal E|+cc-n=f-d$. However, since $\mathrm{Reach}(\mathcal G)$ contains $R$, there is a DFVS of $R$ of size $f-d$ contradicting the minimality of $f$. With Lemma~\ref{lemma:constr}, we can conclude that there is a solution iff $k \geq n + f - cc$.
  \end{proof}
  Specifically, for MinDCRS, finding a solution is therefore equivalent to finding a minimum DFVS which is known to be NP-complete~\cite{Karp1972} and FPT in the size of the set~\cite{DFVSFPT}, giving us the following corollary:
  \begin{corollary}
    MinDCRS is NP-complete, and FPT with respect to the parameter $k'=k+cc-n$.
  \end{corollary}
  
\section{Undirected case: trees}
\label{sec:undirected}

In this section, we study (undirected) MinCRS. First, note that our result from Section \ref{sec:directed} does not hold when the graph we build is undirected. Indeed, consider the case where the request graph is a bidirected path. Here, a minimal directed feedback vertex set would have size $\lfloor \frac n 2 \rfloor$, but the MinCRS problem has a solution with only $n-1$ edges (any path works regardless of the appearance times of the edges). This contradicts Lemma \ref{lemma:fvs} and Theorem \ref{thm:maindirected}.

In graph theory, problems on undirected graphs are generally easier to figure out, but here, temporality induces an orientation on the paths anyway, as any path of length more than two can only be used in at most one direction because of the appearance times of the edges, even if the edges are undirected. 
Still paths of length one, i.e., temporal edges, can now be used both ways. As illustrated in the paragraph above, this possibility may allow for a drastic reduction of the number of edges needed to satisfy the connectivity requirements. How to optimally take advantage of this possibility turns out to be very challenging.

Let $R$ be an instance of MinCRS and $cc$ be the number of connected components of $R$. It is easy to see that any solution needs at least $n-cc$ temporal edges. This section aims at characterizing the cases where there exists a solution with {\it exactly} $n-cc$ temporal edges, that is for each connected component $C$ of $R$, $\G_\downarrow[C]$ is a tree. In the following, we will assume that $R$ is connected and we ask in which case we have a tree solution.

\subsection{Tree representation}

We first show the following lemma on the structure of such a tree solution.

\begin{lemma}
\label{lemma:closed_walk}
Let $R=(V,A)$ be a request graph such that $R$ is connected and there exists a tree solution $\mathcal{T}$. Let $W=(u_0,u_1,\dots,u_{k-1},u_k=u_0)$ be a closed walk of $R$. Let $S=\{u_0,\dots,u_{k-1}\}$. Then, $\mathcal{T}_{\downarrow}[S]$ is connected. 
\end{lemma}

\begin{proof}
Assume by contradiction that $\mathcal{T}_\downarrow[S]$ is not connected. Then, there exists a vertex $w$ such that deleting $w$ from $\mathcal{T}_\downarrow$ creates at least two connected components containing at least one vertex of $S$ each. Let $C_0,\dots,C_\ell$ be the connected components obtained by deleting $w$. For $j\in [0,\ell]$, we denote $(\{w,s_j\},t_j)$ the temporal edge connecting $w$ to $C_j$ in $\mathcal{T}$. Without loss of generality, assume that $t_0<\dots<t_\ell$. As $W$ is a closed walk, we have that there exists $i\in [0,k-1]$ such that $u_i\in C_p$ and $u_{i+1}\in C_q$ with $q<p$. As $(u_i,u_{i+1})\in R$ and the only path from $u_{i}$ to $u_{i+1}$ in $\mathcal{T}_\downarrow$ contains the edge $\{w,s_p\}$ followed by the edge $\{w,s_q\}$, we must have $t_p<t_q$.\qedhere
 
\end{proof}

Note that having the arcs $(u,v)$ and $(v,u)$ in the connected request graph implies that the edge $\{u,v\}$ must be in the footprint of any tree solution. In the following, we will call such an edge $\{u,v\}$ a {\it forced edge}. Another direct consequence of this lemma is the following:

\begin{corollary}
Let $R=(V,A)$ be a request graph such that $R$ is connected and there exists a tree solution $\mathcal{T}$. Let $W$ and $W'$ be two closed walks of $R$ intersecting on vertex set $S$. Then, $\mathcal{T}_{\downarrow}[S]$ is connected. 
\end{corollary} 

Lemma~\ref{lemma:closed_walk} leads to a rather natural definition of what we call a tree representation. 

\begin{definition}A {\it tree representation} of a request graph $R=(V,A)$ is a tree $T$ on same vertex set $V$ and such that $T[S]$ is connected for each closed walk $W$ of $R$ with vertex set $S$. 
\end{definition}
Note that the footprint of a tree solution for a request graph $R$ is always a tree representation of $R$. One could hope that having a tree representation of a request graph $R$ is equivalent to having a tree solution for MinCRS. However, this is not the case. Indeed, consider the request graph $R=(V,A)$ depicted in Figure~\ref{fig:c-ex_tree-rep}. Any tree on $V$ with edges $\{a,b\}$, $\{b,c\}$, $\{c,d\}$, $\{d,e\}$ and an edge between $f$ and any other node is a tree representation of $R$. Now, if a tree solution exists for MinCRS, it must contain the edges $\{a,b\}$, $\{b,c\}$, $\{c,d\}$ and $\{d,e\}$, as they are forced edges. If the last edge is between $f$ and $a$, $b$ or $c$, the time appearance on the edge $\{e,d\}$ must be smaller than the one on $\{d,c\}$ in order to satisfy the request from $e$ to $f$ in $R$. However, the request from $c$ to $e$ in $R$ enforces the time appearance on the edge $\{e,d\}$ to be greater than the one on $\{d,c\}$. By symmetry, we conclude that it is also not possible to have the edge between $f$ and $d$ or $e$. Thus, there is no tree solution for MinCRS. 

    \begin{figure}[!ht]
      \centering
        \begin{tikzpicture}[node distance={20mm}, thick, main/.style = {draw=none, circle, minimum size=0.8cm}]
          \node[main] (a) [] {$a$};
          \node[main] (b) [right = 1cm of a] {$b$};
          \node[main] (c) [right= 1cm of b] {$c$};
          \node[main] (d) [right= 1cm of c] {$d$};
          \node[main] (e) [right = 1cm of d] {$e$};
         \node[main] (f) [below = 1.5cm of c] {$f$};

        \draw[->] (b) to[out=170,in=10] (a) ;
        \draw[->] (a) to[out=-10,in=-170] (b) ;
        \draw[->] (c) to[out=170,in=10] (b) ;
        \draw[->] (b) to[out=-10,in=-170] (c) ;
        \draw[->] (d) to[out=170,in=10] (c) ;
        \draw[->] (c) to[out=-10,in=-170] (d) ;
        \draw[->] (e) to[out=170,in=10] (d) ;
        \draw[->] (d) to[out=-10,in=-170] (e) ;
         \draw[->] (e) -- (f);
        \draw[->] (a) -- (f);
        \draw[->] (c) to[out=30,in=150] (e) ;
        \draw[->] (c) to[out=150,in=30] (a) ;
        \end{tikzpicture}
      \caption{A request graph that admits a tree representation but does not admit a tree solution for the MinCRS problem.}
\label{fig:c-ex_tree-rep}
    \end{figure}
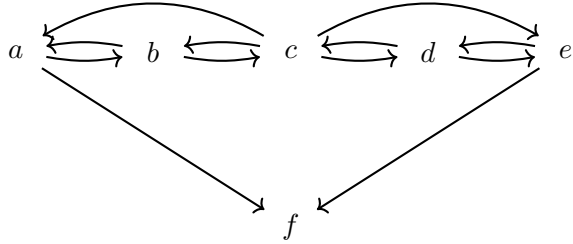

However, we will prove that, when the request graph is strongly connected, having a tree representation is {\it equivalent} to having a tree solution for MinCRS.

\subsection{Hypertree and Helly property}

Before showing our result, we will define and introduce useful concepts and notations. 

\begin{definition}[Hypergraph]
A hypergraph is a pair $H=(V,E)$ where $V$ is a (finite) set of vertices and $E=\{E_1,\dots,E_m\}$ is a set of (distinct) hyperedges. A hyperedge $E_i$ is a nonempty subset of $V$. 
\end{definition} 

Our definition of a tree representation of a request graph is a special case of the concept of hypertree introduced in~\cite{flament1978} (see~\cite{difonzo2025} for a survey on the topic). 

\begin{definition}[Hypertree]
A hypergraph $H=(V,E)$ is a hypertree if it admits a tree $T$ with same vertex set and such that for every hyperedge $E_i$ of $H$, $T[E_i]$ is connected. Such a tree $T$ is called a host tree. 
\end{definition}

Indeed, for a request graph $R=(V,A)$, we can define the hypergraph $H_R$, called the {\it closed walk hypergraph} of $R$, with same vertex set $V$ and for each closed walk $W$ of $R$ on vertex set $S$, $S$ is a hyperedge of $H_R$. Then, a host tree of $H_R$ is a tree representation of $R$ and vice-versa. 

In~\cite{flament1978}, the author proposes a characterisation of hypertrees. Before stating this characterisation, we introduce the definitions needed. 

    \begin{definition}[Helly Property]
        Let $\mathcal{S}=\{S_1,\dots,S_p\}$ be a set of subsets of a given ground set $S$. We say that $\mathcal{S}$ has the Helly property if for all subset $\mathcal{S}'\subseteq \mathcal{S}$:
        
        $$\left( \forall S_i,S_j \in \mathcal{S}' \ S_i\cap S_j\neq \emptyset \right) \Longrightarrow \cap_{S_i\in \mathcal{S}'} S_i \neq \emptyset$$ 
    \end{definition}
    So, for any subset of pairwise non disjoints sets, these sets have an element in common. It is well known for instance that the set of subtrees of a given tree has the Helly property.  
    
    We say that a hypergraph $H=(V,E)$ is {\it Helly} if its set of hyperedges $E$ has the Helly property. This leads us to the following definition. 
    \begin{definition}[walk-Helly]
    We say that a digraph $R$ is {\it walk-Helly} if the set of vertices of the closed walks of $R$ has the Helly property, i.e., any set of pairwise non-vertex-disjoint closed walks all share a common vertex. 
    \end{definition} 
    Equivalently, $R$ is walk-Helly if the closed walk hypergraph of $R$ is Helly.
    
    \begin{example}\label{ex:not-walk-helly}
 
    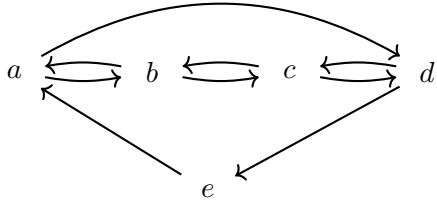
\begin{figure}[!ht]
      \centering
        \begin{tikzpicture}[node distance={20mm}, thick, main/.style = {draw=none, circle, minimum size=0.8cm}]
          \node[main] (a) [] {$a$};
          \node[main] (b) [right = 1cm of a] {$b$};
          \node[main] (c) [right= 1cm of b] {$c$};
          \node[main] (d) [right= 1cm of c] {$d$};
         \node[main] (e) [below left = 1cm and 0.5cm of c] {$e$};

        \draw[->] (b) to[out=170,in=10] (a) ;
        \draw[->] (a) to[out=-10,in=-170] (b) ;
        \draw[->] (c) to[out=170,in=10] (b) ;
        \draw[->] (b) to[out=-10,in=-170] (c) ;
        \draw[->] (d) to[out=170,in=10] (c) ;
        \draw[->] (c) to[out=-10,in=-170] (d) ;
         \draw[->] (d) -- (e);
        \draw[->] (e) -- (a);
        \draw[<-] (d) to[out=150,in=30] (a) ;
        \end{tikzpicture}
\caption{A request graph that is not walk-Helly.}
\label{fig:not-walk-helly}
    \end{figure} 
    Let us consider the request graph $R$ depicted in Figure~\ref{fig:not-walk-helly}. We have that $R$ is not walk-Helly as we have the closed walks $W_1=(a,b,a)$, $W_2=(b,c,d,c,b)$ and $W_3=(a,d,e,a)$ such that $W_1\cap W_2\neq \emptyset$, $W_1\cap W_3\neq \emptyset$ and $W_2\cap W_3\neq \emptyset$ but $W_1\cap W_2\cap W_3= \emptyset$. 
    
    \hfill{} $\Diamond$ 
    \end{example}
    
    We conclude this section with a theorem by Flament that gives a necessary and sufficient condition for a hypergraph to be a hypertree \cite{flament1978}.
    
    Given a hypergraph $H=(V,E)$, its {\it line graph} $L(H)$ is the undirected graph whose vertex set is the set of hyperedges of $H$ and there exists an edge $\{E,E'\}$ in $L(H)$ if and only if the hyperedges $E$ and $E'$ intersect. 
    Recall that an undirected graph is chordal if all cycles of four or more vertices has a chord, that is an edge that is not an edge of the cycle that connects two vertices of the cycle.
    
    \begin{theorem}[Flament, 1978]
    \label{thm:hypertree}
    A hypergraph $H$ is a hypertree if and only if $H$ is Helly and $L(H)$ is chordal.
    \end{theorem}

\subsection{Strongly connected request graph}

The main result of this section is the following. 

\begin{theorem}
\label{thm:main}
Let $R$ be a strongly connected request graph. Then the three following statements are equivalent:
\begin{enumerate}
	\item $R$ is walk-Helly.
	\item There exists a tree representation of $R$.
	\item There exists a tree solution for MinCRS.
\end{enumerate}
Moreover, if such a tree solution exists, we can compute one in polynomial time. 
\end{theorem}

\begin{example} [Example~\ref{ex:not-walk-helly} continued]
Let us consider the request graph of Figure~\ref{fig:not-walk-helly} that is not walk-Helly. 
We can easily see that there exists no tree solution for MinCRS as, no matter where we put the last edge (the edges $(a,b),(b,c)$ and $(c,d)$ are forced), we cannot assign appearance times so that the requests $(a,d)$, $(d,e)$ and $(e,a)$ are satisfied. \hfill{} $\Diamond$ 
\end{example}

We will first show the equivalence between statements 1 and 2, that strongly relies on Theorem~\ref{thm:hypertree}. Note that for this equivalence, the connected request graph $R$ does not need to be strongly connected.

\begin{proposition}
\label{prop:1-2}
Let $R$ be a connected digraph. $R$ is walk-Helly if and only if there exists a tree representation of $R$. 
\end{proposition}

\begin{proof}
Let $H_R$ be the closed walk hypergraph of $R$. 

$\Rightarrow$: We show that $R$ being walk-Helly implies that $L(H_R)$ is chordal. Let $C=(E_1,\dots,E_k,E_1)$ be a cycle in $L(H_R)$ with $k\geq 4$ and such that $C$ does not have a chord. Recall that $E_i$ is a set of vertices that forms a closed walk in $R$ (with a slight abuse of notation, we use $E_i$ to denote the vertex of $L(H_R)$ but also the associated hyperedge in $H_R$). Let $S=\bigcup_{i\in [3,k]} E_i$. By concatenating the closed walks $E_3,\dots,E_k$ (possible as $E_i$ intersects $E_{i+1}$ for each $i\in [3,k-1]$), we obtain a closed walk $W=(u_1,\dots,u_\ell)$ in $R$ with $\{u_1,\dots,u_\ell\}=S$. First, we have $E_1\cap E_2\neq \emptyset$. As $E_1\cap E_k\neq \emptyset$, we have $E_1\cap S\neq \emptyset$ and as $E_2\cap E_3\neq \emptyset$, we have $E_2\cap S\neq \emptyset$. As $H_R$ is Helly, $E_1\cap E_2\cap S \neq \emptyset$. Let $s\in E_1\cap E_2\cap S$. There exists $i\in [3,k]$ such that $s\in E_i$. 
Let us take such a $i$ that is minimal. If $i\neq 3$, we have that $s\in E_2\cap E_i$, and thus there is an edge between $E_i$ and $E_2$ in $L(H_R)$, contradicting the fact that $C$ has no chord. If $i=3$, we have that $s\in E_1\cap E_i$, which again gives us a chord between $E_1$ and $E_i$ in $C$. Thus, we have that $H_R$ is Helly and $L(H_R)$ is chordal. By Theorem~\ref{thm:hypertree}, $H_R$ is a hypertree, that is there exists a tree representation of $R$.  

$\Leftarrow$: We have that $H_R$ is a hypertree. By Theorem~\ref{thm:hypertree}, we thus have that $H_R$ is Helly, that is $R$ is walk-Helly. 
\end{proof}

Let us now prove the equivalence between statements 1 and 3 of Theorem~\ref{thm:main}. First, statement 3 trivially implies statement 1 (here again, the strong connexity of $R$ is not required). 

\begin{proposition}
\label{prop:3-1}
Let $R$ be a connected request graph. If there exists a tree solution for MinCRS, then $R$ is walk-Helly.
\end{proposition}

\begin{proof}
Let $\mathcal{T}$ be such a tree solution. In particular, thanks to Lemma~\ref{lemma:closed_walk}, $\mathcal{T}_\downarrow$ is a tree representation of $R$, that is to say the closed walk hypergraph $H_R$ of $R$ is a hypertree. Thus, by Theorem~\ref{thm:hypertree}, $R$ is walk-Helly. 
\end{proof}

Finally, let us prove that statement 1 implies statement 3, which is the core of the demonstration. 
The proof will consist of three steps:
\begin{itemize}
    \item We introduce the notion of authorized arc and show that adding an authorized arc to a strongly connected request graph that is walk-Helly preserves the walk-Helly property (Lemma~\ref{lemma:authorized}).
    \item We show that, as long as there are not $n-1$ forced edges, we can find an authorized arc that is not in the request graph $R$, if $R$ is walk-Helly (Lemma~\ref{lemma:find}). 
    \item We finally show that, when having $n-1$ forced edges in a request graph that is walk-Helly, we can assign labels to the forced edges in order to obtain a tree solution for MinCRS (Lemma~\ref{lemma:tempo}). 
\end{itemize}

Let $R$ be a connected request graph. In the following, we will write $P(a,b,\overline{c})$ to denote the existence of a path from $a$ to $b$ that does not contain $c$ in $R$, and $W(a,b,\overline{c})$ to denote the existence of a closed walk containing $a$ and $b$ but not $c$ in $R$, for $a,b,c$ vertices of $R$. We say that an arc $(u,v)$ is {\it authorized} if for every vertex $x$ of $R$ we do not have both $W(u,x,\overline{v})$ and $W(v,x,\overline{u})$.\footnote{Note that if an arc $(u,v)$ is not authorized, then the edge $\{u,v\}$ cannot be part of a tree solution. Indeed, having the edge in a tree solution would contradict Lemma~\ref{lemma:closed_walk} as we have $W(u,x,\overline{v})$ and $W(v,x,\overline{u})$ for some $x$.} 

\begin{example}\label{ex:algo}
 Consider the request graph $R=(V,A)$ depicted in Figure~\ref{fig:authorized-arcs}. 
 
    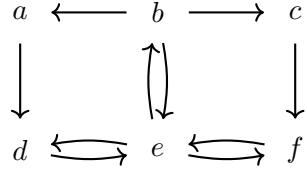
\begin{figure}[!ht]
      \centering
        \begin{tikzpicture}[node distance={20mm}, thick, main/.style = {draw=none, circle, minimum size=0.8cm}]
          \node[main] (a) [] {$a$};
          \node[main] (b) [right = 1cm of a] {$b$};
          \node[main] (c) [right= 1cm of b] {$c$};
          \node[main] (d) [below= 1cm of a] {$d$};
         \node[main] (e) [right = 1cm of d] {$e$};
         \node[main] (f) [right = 1cm of e] {$f$};
        \draw[->] (e) to[out=170,in=10] (d) ;
        \draw[->] (d) to[out=-10,in=-170] (e) ;
        \draw[->] (f) to[out=170,in=10] (e) ;
        \draw[->] (e) to[out=-10,in=-170] (f) ;
        \draw[->] (b) to[out=-80,in=80] (e) ;
        \draw[->] (e) to[out=100,in=-100] (b) ;
        \draw[->] (b) -- (a);
        \draw[->] (b) -- (c);
        \draw[->] (a) -- (d);
        \draw[->] (c) -- (f);
        \end{tikzpicture}
\caption{A strongly connected request graph $R$ that is walk-Helly and whose set of authorized arcs is $\{(a,b),(b,a),(a,d),(d,a),(a,e),(e,a),(c,b),(b,c),(c,e),(e,c),(c,f),(f,c)\}$ plus the arcs associated with forced edges.}
\label{fig:authorized-arcs}
    \end{figure}

 The arc $(a,c)$ is {\it not} authorized: indeed, there is a closed walk $(a,d,e,b,a)$ (so we have $W(a,b,\overline{c})$) and a closed walk $(c,f,e,b,c)$ (so we have $W(c,b,\overline{a})$). As a matter of fact, one can easily see that with the arc $(a,c)$ there would not be any tree solution.   
 
 On the other hand, the arc $(a,b)$ is authorized (as for instance every circuit containing $a$ contains $b$). More generally, the set of authorized arcs of $R$ is $$\{(a,b),(b,a),(a,d),(d,a),(a,e),(e,a),(c,b),(b,c),(c,e),(e,c),(c,f),(f,c)\}$$ \noindent plus the arcs associated with forced edges.\hfill{} $\Diamond$ 
\end{example}
We first show that if $R$ is strongly connected and walk-Helly, then adding an authorized arc to $R$ preserves the property of being walk-Helly. For instance, we have that the strongly connected request graph $R$ of Figure~\ref{fig:authorized-arcs} is walk-Helly and adding the authorized arc $(a,b)$ preserves the walk-Helly property.

\begin{lemma}
\label{lemma:authorized}
Let $R$ be a strongly connected request graph that is walk-Helly, and $(u,v)$ be an authorized arc. Then the graph $R'$ obtained by adding the arc $(u,v)$ to $R$ is walk-Helly.  
\end{lemma}
 
\begin{proof}
Suppose that $R'$ is not walk-Helly, and take $k$ closed walks $A_1,\dots,A_k$ that are pairwise non disjoint but $\cap_{i=1,\dots,k} A_i=\emptyset$. By taking the smallest such $k$, we can assume that any intersection of $k-1$ sets $A_i$ is non empty. Moreover, take such $k$ closed walks $A_1,\dots,A_k$, with $k$ minimal, such that the number of closed walks that use the arc $(u,v)$ is minimized. 

Since the graph $R$ is walk-Helly, then $(u,v)$ must be used in some of the walks, say in walk $A_1$ (so $u$ and $v$ are in $A_1$).

Let us first state an easy property: let $C$ be an intersection of some (at least one) $A_i$. If a vertex $t$ is not in $C$, then for any vertices $a,b$ in $C$ we have $W(a,b,\overline{t})$. Indeed, since $t$ is not in $C$, it is not in one of the $A_i$ (whose intersection is $C$), and we have such a closed walk in $A_i$.

Now we consider different cases.\\

\noindent {\bf Case 1: both $u$ and $v$ belong only to $A_1$ (and no other $A_i$)}.

Let $B=\cap_{i\geq 3} A_i$. Let $P$ be a path from $u$ to $v$ in $R$. Then $P$ must contain a vertex $x$ in $A_2\cap B$ (otherwise we could add $P$ to $A_1$ instead of $(u,v)$ and get that $R$ was not Helly). So we have $P(u,x,\overline{v})$ and $P(x,v,\overline{u})$.

Let $y \in A_1\cap A_2$ and $z\in A_1\cap B$. Considering $A_2$ we have $W(y,x,\overline{u})$ and $W(y,x,\overline{v})$ (since $u$ and $v$ are not in $A_2$). Similarly, considering $B$ and the previous property, since $u,v$ are not in $B$ we have $W(x,z,\overline{u})$ and $W(x,z,\overline{v})$. In other words, we can walk 'freely' (without going through $u$ or $v$) between $x$, $y$ and $z$.
 
Now let us look at $A_1$. 
\begin{itemize}
	\item If we have a closed walk containing $y$ and $z$ but not the arc $(u,v)$, then $R$ is not Helly, contradiction. So we must use $(u,v)$ either to go from $y$ to $z$ or to go from $z$ to $y$.
	\item If we must use $(u,v)$ to go from $y$ to $z$, then we have $P(y,u,\overline{v})$ and $P(v,z,\overline{u})$. But then 
we have $P(x,u,\overline{v})$ - so $W(u,x,\overline{v})$ -  and $P(v,x,\overline{u})$ - so $W(v,x,\overline{u})$, and $(u,v)$ would not have been authorized.
	\item  The case where we must use $(u,v)$ to go from $z$ to $y$ is completely similar.\\
\end{itemize}
 
\noindent {\bf Case 2: there exists $i\geq 2$ such that $v$ belongs to $A_i$ but not $u$}.

We set w.l.o.g. $i=2$, i.e. $v\in A_2$ and $u\not\in A_2$. There exists $A_i$ such that $v\not\in A_i$. Let $C$ be the intersections of all $A_i$ that do not contain $v$. 

Let $P$ be a path from $u$ to $v$ in $R$. Suppose it does not contain any vertex in $A_2\cap C$. Then, we could add $P$ to $A_1$ instead of $(u,v)$ and this contradicts our choice of $A_1,\dots, A_k$.  

So $P$ must contain a vertex $x\in A_2\cap C$. We have $P(u,x,\overline{v})$ and $P(x,v,\overline{u})$.

Considering $A_2$, we directly have $W(x,v,\overline{u})$. 

Now let us consider two cases:
\begin{itemize}
	\item If $u$ is in $C$, then we directly have $W(u,x,\overline{v})$, contradicting the fact that $(u,v)$ is authorized.
	\item If $u$ is not in $C$,
 let $y\in A_1\cap C$. As $v$ is not in $C$, we have $W(x,y,\overline{v})$. 

Let us look at the journeys in $A_1$ (i.e. using only vertices in $A_1$). There is a path from $v$ to $y$ that does not use $(u,v)$. If there is a path from $y$ to $v$ that does not use $(u,v)$, then $(u,v)$ is useless in $A_1$ and this contradicts our choice of $A_1,\dots, A_k$ (take as $A_1$ a closed walk containing $v$ and $y$: $A_1$ still intersects all $A_i$ as they either contain $v$ or $y$). Otherwise all paths from $y$ to $v$ use $(u,v)$, meaning that we have $P(y,u,\overline{v})$. Then we have  $P(x,u,\overline{v})$, so $W(u,x,\overline{v})$.

Again, impossible since $(u,v)$ is authorized.\\
\end{itemize}

\noindent {\bf Case 3: there exists $i\geq 2$ such that $u$ belongs to $A_i$ but not $v$}.

We fix w.l.o.g. $i=2$, i.e., $u\in A_2$ but $v\not\in A_2$. Symmetrically as in the previous case, let $C$ be the intersection of all $A_i$ that do not contain $u$.
As previously, we have some $x\in C\cap A_2$ with $P(u,x,\overline{v})$ and $P(x,v,\overline{u})$. 
 
By considering $A_2$, we have $W(u,x,\overline{v})$. 
\begin{itemize}
	\item If $v$ is in $C$, then we directly have (thanks to $C$) $W(x,v,\overline{u})$, and $(u,v)$ would not have been authorized.
	\item If $v\not\in C$. Let $y$ in $A_1\cap C$. We have $W(x,y,\overline{u})$. Let us look at the journeys in $A_1$ (i.e. using only vertices in $A_1$). We have a path from $y$ to $u$ that does not use $(u,v)$. If there is a path from $u$ to $y$ that does not use $(u,v)$ as well, then $(u,v)$ is useless in $A_1$ and this contradicts our choice of $A_1,\dots, A_k$. Then, every path from $u$ to $y$ uses $(u,v)$, meaning that we have $P(v,y,\overline{u})$. Then we have $P(v,x,\overline{u})$, and then $W(v,x,\overline{u})$. Contradiction again with the fact that $(u,v)$ is authorized. \\
\end{itemize}

\noindent  {\bf Case 4: there exists $i\geq 2$ such that both $u\in A_i$ and $v\in A_i$.}

We fix w.l.o.g. $i=2$, i.e., $u,v\in A_2$. Let $B=\cap_{i\geq 3} A_i$. Note that $u$ and $v$ are not in $B$.

Let $P$ be a path from $u$ to $v$ in $R$. $P$ must intersect $B$ (otherwise we could add $P$ to $A_1$ instead of $(u,v)$ and this contradicts our choice of $A_1,\dots, A_k$). Then we have some $x\in B$ with $P(u,x,\overline{v})$ and $P(x,v,\overline{u})$. 

Let $y\in A_1\cap B$. As $u$ and $v$ are not in $B$, we have a closed walk containing $x,y$ but not $u$ and a closed walk containing $x,y$ but not $v$. 

Let us look at the journeys in $A_1$ (i.e. using only vertices in $A_1$). There is a path from $v$ to $y$ without $(u,v)$. Then:
\begin{itemize}
\item If there is a path from $y$ to $v$ without $u$: there is in $A_1$ a closed walk containing $y$ and $v$ that does not contain $(u,v)$. This contradicts our choice of $A_1,\dots, A_k$. 
\item Otherwise, every path from $y$ to $v$ contains $u$. So we have $P(y,u,\overline{v})$, and then $W(x,u,\overline{v})$. Then: either there is $P(v,y,\overline{u})$ and we get $W(v,x,\overline{u})$ (and $(u,v)$ would not have been authorized), or we have in $A_1$ a path $P(u,y,\overline{v})$. But in this latter case we have in $A_1$ a closed walk containing $u$ and $y$ but not $v$. Then we restrict $A_1$ to be this closed walk, and this contradicts our choice of $A_1,\dots, A_k$. \qedhere

\end{itemize}
\end{proof}

The general idea of the algorithm that builds a tree solution if it exists is to find an authorized arc $(u,v)$ that is not in $R$ and add the arcs $(u,v)$ and $(v,u)$ to $R$ until there are $n-1$ forced edges forming a tree. Note that we can add both the arcs $(u,v)$ and $(v,u)$ as if $(u,v)$ is authorized, then $(v,u)$ is also authorized, and adding the arc $(u,v)$ to $R$ does not change the fact that $(v,u)$ is authorized by definition of an authorized arc. The following lemma shows that, as long as we do not have $n-1$ forced edges forming a tree, we can find an authorized arc $(u,v)$ when $R$ is walk-Helly. 

\begin{example} [Example~\ref{ex:algo} continued]
Let us consider our running example of Figure~\ref{fig:authorized-arcs}. We can add the authorized arc $(a,b)$ to $R$ to obtain the request graph $R'$ whose set of authorized arcs becomes $\{(c,b),(b,c),(c,e),(e,c),(c,f),$ $(f,c)\}$ plus the arcs associated with forced edges. We can see that the set of authorized arcs has decreased ($(a,d),(d,a),(a,e)$ and $(e,a)$ are no longer authorized) but we can still find an authorized arc. \hfill{} $\Diamond$
\end{example}

\begin{lemma}
\label{lemma:find}
Let $R$ be a strongly connected request graph that is walk-Helly and such that there are not $n-1$ forced edges forming a tree. Then, there exists an authorized arc $(u,v)$ that is not in $R$. 
\end{lemma}

\begin{proof}
Let $T$ be a tree representation of $R$ (exists as $R$ is walk-Helly and by Proposition~\ref{prop:1-2}) and let $\{u,v\}$ be an edge of $T$ that is not a forced edge in $R$. We have that $(u,v)$ or $(v,u)$ (or both) are not in $R$. Without loss of generality, let us assume that $(u,v)$ is not in $R$ and let us show that the arc $(u,v)$ is authorized. If $(u,v)$ is not authorized, there exists a vertex $x$ such that $W(u,x,\bar{v})$ and $W(v,x,\bar{u})$ in $R$. Consider the request graph $R'$ which is $R$ to which we add the arcs $(s,t)$ and $(t,s)$ for each edge $\{s,t\}$ of $T$ (if not already in $R$). We have that $T$ is a tree representation of $R'$, so $R'$ is walk-Helly by Proposition~\ref{prop:1-2}. However, we have $W(u,x,\bar{v})$ and $W(v,x,\bar{u})$ in $R'$ and also the closed walk consisting of the arc $(u,v)$ and the arc $(v,u)$ in $R'$, contradicting the fact that $R'$ is walk-Helly. 
\end{proof}

We finally show that, when we have $n-1$ forced edges forming a tree, we can assign labels to the tree in order to obtain a solution for MinCRS if the request graph $R$ is walk-Helly. An illustration of the construction is given in an example after the proof. In Figure~\ref{fig:temp} the left part represents the request graph, and the right part represents a tree solution with labels (the middle part illustrates part of the proof).

\begin{lemma}
\label{lemma:tempo}
Let $R=(V,A)$ be an instance of MinCRS such that $R$ is connected and has $n-1$ forced edges forming a tree representation $T$ of $R$. Then, each edge of the tree representation can be given an appearance time so that the obtained temporal graph is a solution.  
\end{lemma}

\begin{proof}
For a forced edge $\{u,v\}$, we denote $t_{\{u,v\}}$ the appearance time that we will give to this edge. First, observe that each arc $(u,v)\in R$ such that $\{u,v\}$ is not a forced edge enforces an order on the appearance times of the edges on the path from $u$ to $v$ in $T$. More precisely, let $P=(u_0=u,u_1,\dots,u_k=v)$ be the unique path from $u$ to $v$ in $T$ and let $e_i=\{u_i,u_{i+1}\}$ for $i\in [0,k-1]$. Then, we must have $t_{e_0}<t_{e_1}<\dots<t_{e_{k-1}}$. 

Let us construct the following digraph $D$. For each forced edge $e$ of $R$, we have a vertex $v_e$ in $D$. For each arc $(u,v)\in R$ such that $\{u,v\}$ is not a forced edge, let $P=(u_0=u,u_1,\dots,u_k=v)$ be the path from $u$ to $v$ in $T$ and $e_i=\{u_i,u_{i+1}\}$. Then, we put an arc from $v_{e_i}$ to $v_{e_{i+1}}$ in $D$ for $i\in [0,k-2]$. Note that having an arc from $v_{e}$ to $v_{e'}$ in $D$ means that $e$ and $e'$ are sharing an endpoint in $T$ and that we must have $t_e<t_{e'}$. If there is no closed walk in $D$, then assigning labels to the forced edges by following a topological order of $D$ ensures that the obtained temporal graph is a solution. 

Now, let us assume that there exists a closed walk in $D$. Let $W=(v_{e_0},v_{e_1},\dots,v_{e_k}=v_{e_0})$ be such a closed walk with minimal length. 
Let us show that all forced edges $e_i$ for $i\in [0,k-1]$ share a common endpoint in $T$. First, we have that $e_0$ and $e_1$ shares an endpoint, that we will denote $s$. Now, let us assume that all forced edges $e_i$ for $i\in [0,j]$ have $s$ as one endpoint, for $1\leq j \leq k-3$. We have that $e_{j+1}$ and $e_j$ have a common endpoint. Assume that this is not $s$. Then there exists a path in $T$ consisting of the edges $e_{j+1}$, $e_j$ and $e_{0}$ in this order. However, we have a path from $v_{e_{j+1}}$ to $v_{e_0}$ in $D$ and as there exists an arc from a vertex associated with a forced edge to another in $D$ only between forced edges sharing an endpoint, we have that $v_{e_j}$ should be in every path from  $v_{e_{j+1}}$ to $v_{e_0}$ in $D$, contradicting the minimality of $W$. Finally, $e_{k-1}$ shares an endpoint with $e_0$ and also with $e_{k-2}$ in $T$, which is thus $s$. 

Let $e_i=\{s,w_i\}$ for $i\in [0,k-1]$ and let $T_i$ be the subtree containing $w_i$ obtained when deleting vertex $s$ of $T$. As $(v_{e_i},v_{e_{i+1}})$ is an arc of $D$ for $i\in [0,k-1]$, there exists an arc in the request graph $R$ from a vertex $s_i\in T_i$ to a vertex $t_{i+1}\in T_{i+1}$. Also, there exists a path in $R$ from $t_i$ to $s_i$ for $i\in [1,k-1]$ (resp. from $t_k$ to $s_0$) obtained by following the path from $t_i$ to $s_i$ in $T_i$ (resp. from $t_k$ to $s_0$ in $T_0$). Concatenating those paths with the arcs from $s_i$ to $t_{i+1}$ gives a closed walk in $R$ that contains all $s_i$ and $t_i$ (and potentially other vertices) but that does not contain $s$. This contradicts the fact that $T$ is a tree representation of $R$. 
\end{proof}

From this proof, we can directly derive a polynomial time algorithm, called \textsc{Labelisation}, such that, given a connected request graph $R$ with $n-1$ forced edges, it assigns appearance times to the forced edges so that the resulting temporal graph is a solution for MinCRS, or outputs NO if it is not possible. 

\begin{example}[Example~\ref{ex:algo} continued]
Let us go back to our running example (Figure~\ref{fig:authorized-arcs}). Assume that after choosing to add the authorized arc $(a,b)$, we chose to add the authorized arc $(c,b)$. We then obtain a request graph with $n-1$ forced edges forming a tree representation (see Figure~\ref{fig:temp}, left). In the middle of Figure~\ref{fig:temp}, the associated digraph $D$ described in the proof of Lemma~\ref{lemma:tempo} is depicted. For instance the arc $(a,d)$ in the request graph gives the arcs $(v_{\{a,b\}},v_{\{b,e\}})$ and $(v_{\{b,e\}},v_{\{d,e\}})$ in $D$. We then construct a tree solution for MinCRS by giving appearance times to the forced edges according to the topological order of $D$ (Figure~\ref{fig:temp}, right).\hfill{} $\Diamond$
\end{example}

We can now conclude the proof of Theorem~\ref{thm:main} by showing that statement 1 implies statement 3. 

    \begin{figure}[!ht]
      \centering
        \begin{tikzpicture}[node distance={20mm}, thick, main/.style = {draw=none, circle, minimum size=0.55cm}]
          \node[main] (a) [] {$a$};
          \node[main] (b) [right = 0.5cm of a] {$b$};
          \node[main] (c) [right= 0.5cm of b] {$c$};
          \node[main] (d) [below= 0.5cm of a] {$d$};
         \node[main] (e) [right = 0.5cm of d] {$e$};
         \node[main] (f) [right = 0.5cm of e] {$f$};
        \draw[->] (b) to[out=170,in=10] (a) ;
        \draw[->] (a) to[out=-10,in=-170] (b) ;
        \draw[->] (c) to[out=170,in=10] (b) ;
        \draw[->] (b) to[out=-10,in=-170] (c) ;
        \draw[->] (e) to[out=170,in=10] (d) ;
        \draw[->] (d) to[out=-10,in=-170] (e) ;
        \draw[->] (f) to[out=170,in=10] (e) ;
        \draw[->] (e) to[out=-10,in=-170] (f) ;
        \draw[->] (b) to[out=-80,in=80] (e) ;
        \draw[->] (e) to[out=100,in=-100] (b) ;
        \draw[->] (a) -- (d);
        \draw[->] (c) -- (f);
 
        \node[main] (v1) [above right = 0.1cm and 0.03cm of c] {};
        \node[main] (v2) [below = 2 cm of v1] {};
        \node[main] (v3) [right = 6.35 cm of v1] {};
        \node[main] (v4) [below = 2cm of v3] {};
        \draw[-] (v1) -- (v2);
        \draw[-] (v3) -- (v4);
        
        \node[main] (vab) [below right = 0.1 cm and 0.5cm of c] {\small{$v_{\{a,b\}}$}};
          \node[main] (vbc) [right = 0.1cm of vab] {\small{$v_{\{b,c\}}$}};
          \node[main] (vbe) [right= 0.3cm of vbc] {\small{$v_{\{b,e\}}$}};
          \node[main] (vde) [right= 0.3cm of vbe] {\small{$v_{\{d,e\}}$}};
         \node[main] (vef) [right = 0.1cm of vde] {\small{$v_{\{e,f\}}$}};

        \draw[->] (vab) to[out=30,in=150] (vbe) ;
        \draw[->] (vbc) -- (vbe) ;
        \draw[->] (vbe) -- (vde) ;
        \draw[->] (vbe) to[out=-30,in=-150] (vef) ;

          \node[main] (a2) [right = 7cm of c] {$a$};
          \node[main] (b2) [right = 0.5cm of a2] {$b$};
          \node[main] (c2) [right= 0.5cm of b2] {$c$};
          \node[main] (d2) [below= 0.5cm of a2] {$d$};
         \node[main] (e2) [right = 0.5cm of d2] {$e$};
         \node[main] (f2) [right = 0.5cm of e2] {$f$};

       \draw[-] (a2) -- node[midway, above] {1} (b2) ;
       \draw[-] (b2) -- node[midway, above] {2} (c2) ;
       \draw[-] (b2) -- node[midway, left] {3} (e2) ;
       \draw[-] (d2) -- node[midway, below] {4} (e2) ;
       \draw[-] (e2) -- node[midway, below] {5} (f2) ;

        \end{tikzpicture}
\caption{On the left, a request graph with $n-1$ forced edges forming a tree representation. In the middle, a topological order of the digraph described in the proof of Lemma~\ref{lemma:tempo}. On the right, a tree solution for MinCRS.}
\label{fig:temp}
    \end{figure}
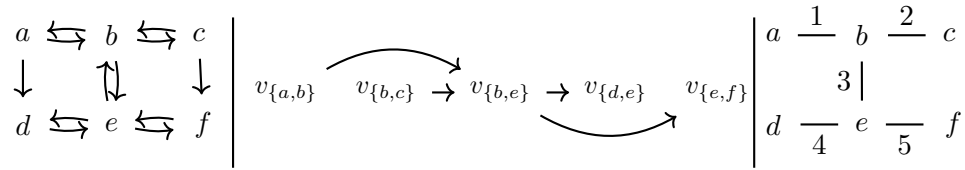

\begin{proposition}
Let $R$ be a strongly connected request graph. If $R$ is walk-Helly, then there exists a tree solution for MinCRS. Moreover, Algorithm~\ref{alg:tree_sol} allows for the computation of such a tree solution if it exists in polynomial time.
\end{proposition}

\begin{algorithm}[h]
	\caption{}
	\label{alg:tree_sol}
	\begin{algorithmic}[1]
	\Require A strongly connected request graph $R$
	\Ensure NO if there is no tree solution for MinCRS, a tree solution otherwise
	
	\While{there exists an authorized arc $(u,v)$ that is not in $R$}
	
	    \State Add $(u,v)$ and $(v,u)$ to $R$ (if not already in it) 
	
	\EndWhile
	
	\If{there are $n-1$ forced edges in $R$ forming a tree $T$}
	    \State \textsc{Labelisation}$(R)$
	\Else 
	    \State Output NO
	\EndIf 
	
\end{algorithmic}
\end{algorithm}
\begin{proof}
Let $R$ be a strongly connected request graph that is walk-Helly. If there are $n-1$ forced edges, then Lemma~\ref{lemma:tempo} shows that we can assign appearance times to the edges of the tree formed by the forced edges so that it is a tree solution for MinCRS. If there are not yet $n-1$ forced edges, we can find an authorized arc $(u,v)$ not in $R$ by Lemma~\ref{lemma:find}, and add the arcs $(u,v)$ and $(v,u)$ in $R$ while maintaining the walk-Helly property by Lemma~\ref{lemma:authorized}. 
Note that if $R$ is not walk-Helly, then there exists no tree solution for MinCRS by Proposition~\ref{prop:3-1} and the algorithm will return NO either because there are not $n-1$ forced edges after adding the authorized arcs to $R$, or because the labelisation will fail. 
Thus, Algorithm~\ref{alg:tree_sol} computes a tree solution if it exists in polynomial time. 
\end{proof}

\section{Conclusion}\label{sec:conclusion}

The problem of Connectivity Request Satisfaction that we introduced turned out to be very intriguing and our work raise a lot of open questions. We would like to conclude this paper by presenting some of them that we believe to be of strong interest:
\begin{itemize}
\item Our results in the undirected case hold when the request graph is strongly connected. Thus, the complexity of determining if there exists a tree solution for non-strongly connected instances of MinCRS is still open. Note that all the connected components of a reachability graph have to be strongly connected (since all the edges of the footprint of a temporal graph involve a pair of opposite arc in its reachability graph). Thus, if there exists a tree solution for MinCRS, its reachability graph has to be walk-Helly (and has to contain all the arcs of the request graph). Thus, whether a non-strongly connected instance admits a tree-solution is equivalent to whether we can make it strongly connected by adding arcs to it while keeping it walk-Helly.
\item One of the main question is also the complexity of MinCRS in the general undirected case. While we strongly believe the problem to be NP-complete, this is not a direct consequence of any of the results we were able to achieve. It would also be interesting to connect this problem to already known structural graph parameters like we were able to do with the directed feedback vertex set in the directed case.
\item We know that any connected request graph requires at least $n-1$ edges to be satisfied and we tried to characterize the request graphs for which $n-1$ is enough. On the opposite end, we know that $2n-4$ edges are enough to satisfy any request graph and it would be interesting to characterize the graphs for which it is necessary. This would also probably be an interesting result in gossip theory.
\item Finally, we conjecture that every instance of MinCRS admits a simple optimal solution i.e. an optimal solution where an edge between the same two vertices cannot appear at several different times. This is clear when we look for a solution of cost $n-cc$ and is also proven true when the optimal solution has $2n-4$ edges but is not known for intermediate values. Being able to restrict the search to simple solutions would probably be very helpful for future work on this topic.
\end{itemize}

\bibliography{ref}

@article{flament1978,
title = {Hypergraphes arborés},
journal = {Discrete Mathematics},
volume = {21},
number = {3},
pages = {223-227},
year = {1978},
author = {Claude Flament},
}

@article{DFVSFPT,
  author       = {Jianer Chen and
                  Yang Liu and
                  Songjian Lu and
                  Barry O'Sullivan and
                  Igor Razgon},
  title        = {A fixed-parameter algorithm for the directed feedback vertex set problem},
  journal      = {J. {ACM}},
  volume       = {55},
  number       = {5},
  pages        = {21:1--21:19},
  year         = {2008},
}

@misc{difonzo2025,
      title={Hypertrees and their host trees: a survey}, 
      author={Pablo De Caria Di Fonzo},
      year={2025},
      eprint={2504.15570},
      archivePrefix={arXiv},
}

@inproceedings{Karp1972,
  author       = {Richard M. Karp},
  title        = {Reducibility Among Combinatorial Problems},
  booktitle    = {Proceedings of a symposium on the Complexity of Computer Computations},
  series       = {The {IBM} Research Symposia Series},
  pages        = {85--103},
  year         = {1972},
}

@article{Hajnal1972, title={A Cure for the Telephone Disease}, volume={15}, DOI={10.4153/CMB-1972-081-0}, number={3}, journal={Canadian Mathematical Bulletin}, author={Hajnal, A. and Milner, E. C. and Szemerédi, E.}, year={1972}, pages={447–450}}

@article{Bumby1981,
  title={A Problem with Telephones},
  author={R. Bumby},
  journal={Siam Journal on Algebraic and Discrete Methods},
  year={1981},
  volume={2},
  pages={13-18}
}

@article{Hedetniemi1988,
author = {Hedetniemi, Sandra M. and Hedetniemi, Stephen T. and Liestman, Arthur L.},
title = {A survey of gossiping and broadcasting in communication networks},
journal = {Networks},
volume = {18},
number = {4},
pages = {319-349},
year = {1988}
}

@article{Gobel1991,
author = {G\"{o}bel, F. and Cerdeira, J. Orestes and Veldman, H. J.},
title = {Label-connected graphs and the gossip problem},
year = {1991},
volume = {87},
number = {1},
journal = {Discrete Math.},
pages = {29–40},
}

@article{dreyfus1969,
 author = {Stuart E. Dreyfus},
 journal = {Operations Research},
 number = {3},
 pages = {395--412},
 title = {{An Appraisal of Some Shortest-Path Algorithms}},
 volume = {17},
 year = {1969}
}

@ARTICLE{wu2016,
  author={Wu, Huanhuan and Cheng, James and Ke, Yiping and Huang, Silu and Huang, Yuzhen and Wu, Hejun},
  journal={IEEE Transactions on Knowledge and Data Engineering}, 
  title={{Efficient Algorithms for Temporal Path Computation}}, 
  year={2016},
  volume={28},
  number={11},
  pages={2927-2942},
}

@article{buixuan2003,
  author       = {Binh-Minh Bui-Xuan and
                  Afonso Ferreira and
                  Aubin Jarry},
  title        = {{Computing Shortest, Fastest, and Foremost Journeys in Dynamic Networks}},
  journal      = {Int. J. Found. Comput. Sci.},
  volume       = {14},
  number       = {2},
  pages        = {267--285},
  year         = {2003},
}

@article{othon2016,
title = {Traveling salesman problems in temporal graphs},
journal = {Theoretical Computer Science},
volume = {634},
pages = {1-23},
year = {2016},
author = {Othon Michail and Paul G. Spirakis},
}

@article{bumpus2023,
author = {Bumpus, Benjamin Merlin and Meeks, Kitty},
title = {Edge Exploration of Temporal Graphs},
year = {2022},
volume = {85},
number = {3},
journal = {Algorithmica},
pages = {688–716},
}

@inproceedings{Bellitto2026,
  author       = {Thomas Bellitto and
                  Johanne Cohen and
                  Bruno Escoffier and
                  Minh{-}Hang Nguyen and
                  Mika{\"{e}}l Rabie},
  title        = {Canadian Traveler Problems in Temporal Graphs},
  booktitle    = {Graph-Theoretic Concepts in Computer Science},
  series       = {Lecture Notes in Computer Science},
  pages        = {33--47},
  year         = {2025},
}

@article{erdos1960,
  title={Graphs with prescribed degrees of vertices},
  author={Erd{\H{o}}s, Paul and Gallai, Tibor},
  journal={Mat. Lapok},
  volume={11},
  pages={264--274},
  year={1960}
}

@article{hakimi1965,
  title={Distance matrix of a graph and its realizability},
  author={S. Louis Hakimi and S. S. Yau},
  journal={Quarterly of Applied Mathematics},
  year={1965},
  volume={22},
  pages={305-317},
}

@article{klobas2024,
title = {Temporal graph realization from fastest paths},
author = {Nina Klobas and George B. Mertzios and Hendrik Molter and Paul G. Spirakis},
journal = {Theoretical Computer Science},
volume = {1056},
pages = {115508},
year = {2025},
}

@InProceedings{erlebach2024,
  author =	{Erlebach, Thomas and Morawietz, Nils and Wolf, Petra},
  title =	{{Parameterized Algorithms for Multi-Label Periodic Temporal Graph Realization}},
  booktitle =	{3rd Symposium on Algorithmic Foundations of Dynamic Networks (SAND 2024)},
  pages =	{12:1--12:16},
  series =	{Leibniz International Proceedings in Informatics (LIPIcs)},
  year =	{2024},
  volume =	{292},
}

@inproceedings{mertzios2025,
  author       = {George B. Mertzios and
                  Hendrik Molter and
                  Nils Morawietz and
                  Paul G. Spirakis},
  title        = {Realizing Temporal Transportation Trees},
  booktitle    = {Graph-Theoretic Concepts in Computer Science - 51st International
                  Workshop, {WG} 2025, Otzenhausen, Germany, June 11-13, 2025, Revised
                  Selected Papers},
  series       = {Lecture Notes in Computer Science},
  pages        = {390--404},
  year         = {2025},
}

@inproceedings{mertzios2025-2,
  author       = {George B. Mertzios and
                  Hendrik Molter and
                  Nils Morawietz and
                  Paul G. Spirakis},
  title        = {Temporal Graph Realization with Bounded Stretch},
  booktitle    = {50th International Symposium on Mathematical Foundations of Computer
                  Science, {MFCS} 2025, Warsaw, Poland, August 25-29, 2025},
  series       = {LIPIcs},
  pages        = {75:1--75:19},
  year         = {2025},
}

@inproceedings{erlebach2025,
  author       = {Thomas Erlebach and
                  Othon Michail and
                  Nils Morawietz},
  title        = {Recognizing and Realizing Temporal Reachability Graphs},
  booktitle    = {33rd Annual European Symposium on Algorithms, {ESA} 2025, Warsaw,
                  Poland, September 15-17, 2025},
  series       = {LIPIcs},
  pages        = {93:1--93:18},
  year         = {2025},
}

@InProceedings{meusel2025,
  author =	{Meusel, Julia and M\"{u}ller-Hannemann, Matthias and Reinhardt, Klaus},
  title =	{{Directed Temporal Tree Realization for Periodic Public Transport: Easy and Hard Cases}},
  booktitle =	{25th Symposium on Algorithmic Approaches for Transportation Modelling, Optimization, and Systems (ATMOS 2025)},
  pages =	{3:1--3:22},
  series =	{Open Access Series in Informatics (OASIcs)},
  year =	{2025},
  volume =	{137},
}

@inproceedings{cauvi2025,
  author       = {Justine Cauvi and
                  Nils Morawietz and
                  Laurent Viennot},
  title        = {Foremost, Fastest, Shortest: Temporal Graph Realization Under Various
                  Path Metrics},
  booktitle    = {43rd International Symposium on Theoretical Aspects of Computer Science,
                  {STACS} 2026, Grenoble, France, March 9-13, 2026},
  series       = {LIPIcs},
  pages        = {24:1--24:19},
  year         = {2026},
}

@InProceedings{casteigts2025,
  author =	{Casteigts, Arnaud and D\"{o}ring, Michelle and Morawietz, Nils},
  title =	{{Realization of Temporally Connected Graphs Based on Degree Sequences}},
  booktitle =	{36th International Symposium on Algorithms and Computation (ISAAC 2025)},
  pages =	{17:1--17:18},
  series =	{Leibniz International Proceedings in Informatics (LIPIcs)},
  year =	{2025},
  volume =	{359},
}

@article{akrida2017,
  author       = {Eleni C. Akrida and
                  Leszek Gasieniec and
                  George B. Mertzios and
                  Paul G. Spirakis},
  title        = {The Complexity of Optimal Design of Temporally Connected Graphs},
  journal      = {Theory Comput. Syst.},
  volume       = {61},
  number       = {3},
  pages        = {907--944},
  year         = {2017},
}

@article{klobas2024-2,
  author       = {Nina Klobas and
                  George B. Mertzios and
                  Hendrik Molter and
                  Paul G. Spirakis},
  title        = {The complexity of computing optimum labelings for temporal connectivity},
  journal      = {J. Comput. Syst. Sci.},
  volume       = {146},
  pages        = {103564},
  year         = {2024},
}

@inproceedings{mertzios2013,
  author       = {George B. Mertzios and
                  Othon Michail and
                  Ioannis Chatzigiannakis and
                  Paul G. Spirakis},
  title        = {Temporal Network Optimization Subject to Connectivity Constraints},
  booktitle    = {Automata, Languages, and Programming - 40th International Colloquium,
                  {ICALP} 2013, Riga, Latvia, July 8-12, 2013, Proceedings, Part {II}},
  series       = {Lecture Notes in Computer Science},
  volume       = {7966},
  pages        = {657--668},
  year         = {2013},
}

@article{enright2021,
  author       = {Jessica A. Enright and
                  Kitty Meeks and
                  George B. Mertzios and
                  Viktor Zamaraev},
  title        = {Deleting edges to restrict the size of an epidemic in temporal networks},
  journal      = {J. Comput. Syst. Sci.},
  volume       = {119},
  pages        = {60--77},
  year         = {2021},
}

@article{deligkas2022,
  author       = {Argyrios Deligkas and
                  Igor Potapov},
  title        = {Optimizing reachability sets in temporal graphs by delaying},
  journal      = {Inf. Comput.},
  volume       = {285},
  number       = {Part},
  pages        = {104890},
  year         = {2022},
}

@article{Kempe2002,
  author       = {David Kempe and
                  Jon M. Kleinberg and
                  Amit Kumar},
  title        = {Connectivity and Inference Problems for Temporal Networks},
  journal      = {J. Comput. Syst. Sci.},
  volume       = {64},
  number       = {4},
  pages        = {820--842},
  year         = {2002},
}

@InProceedings{kurita_et_al:LIPIcs.SAND.2025.9,
  author =	{Kurita, Kazuhiro and Marino, Andrea and Schoeters, Jason and Uno, Takeaki},
  title =	{{Spanner Enumeration for Temporal Graphs}},
  booktitle =	{4th Symposium on Algorithmic Foundations of Dynamic Networks (SAND 2025)},
  pages =	{9:1--9:21},
  series =	{Leibniz International Proceedings in Informatics (LIPIcs)},
  year =	{2025},
  volume =	{330},
}

@InProceedings{bellitto_et_al:LIPIcs.SAND.2025.3,
  author =	{Bellitto, Thomas and Bouton Popper, Jules and Escoffier, Bruno},
  title =	{{Temporal Connectivity Augmentation}},
  booktitle =	{4th Symposium on Algorithmic Foundations of Dynamic Networks (SAND 2025)},
  pages =	{3:1--3:16},
  series =	{Leibniz International Proceedings in Informatics (LIPIcs)},
  year =	{2025},
  volume =	{330},
}

\end{document}